\documentclass[12pt]{article}

\usepackage{amssymb}
\usepackage{amsthm}

\newtheorem {theorem} {Theorem}

\newtheorem {proposition} [theorem]{Proposition}

\newtheorem {remark} [theorem]{\bf Remark}

\title{Periodic orbits in analytically perturbed Poisson systems}

\author{Isaac A. Garc\'{\i}a$^{\ 1,*}$ and Benito Hern\'andez-Bermejo$^{\ 2}$}

\date{$^{\ (1)}$ {\small Departament de Matem\`atica. Universitat de
Lleida. \\ Avda. Jaume II, 69. 25001 Lleida, Spain. 
\\ E-mail: {\tt garcia@matematica.udl.cat}}
\\ $ $ \\
$^{\ (2)}$ {\small Departamento de F\'{\i}sica. Universidad Rey Juan Carlos.
\\ Calle Tulip\'{a}n S/N. 28933--M\'{o}stoles--Madrid, Spain. 
\\ E-mail: {\tt benito.hernandez@urjc.es}}}

\begin{document}

\maketitle

\begin{abstract}
Analytical perturbations of a family of finite-dimensional Poisson systems are considered. It is shown that the family  is analytically orbitally conjugate in $U \subset \mathbb{R}^n$ to a planar harmonic oscillator defined on the  symplectic leaves. As a consequence, the perturbed vector field can be transformed in the domain $U$ to the Lagrange standard form. On the latter, use can be made of averaging theory up to second order to study the existence, number and bifurcation phenomena of periodic orbits. Examples are given ranging from harmonic oscillators with a potential and Duffing oscillators, to a kind of zero-Hopf singularity analytic normal form.
\end{abstract}


\mbox{}


\mbox{}


\vfill

\footnoterule

\noindent $\:\: ^*$ {\small Corresponding author. Telephone: (+34) 973702728. Fax: (+34) 973702702.}

\pagebreak

\section{Introduction}

Finite-dimensional Poisson systems (e.g. see \cite{wei1,olv1} and references therein for an overview and a historical discussion) are present in most fields of physics (including mechanics, electromagnetism, plasma physics, optics, dynamical systems theory) and applied mathematics, as well as closely related areas such as engineering (for instance in control theory), population dynamics, etc. Actually, the description of a given dynamical system in terms of a Poisson structure has implied the development of a number of applied mathematical tools for the obtainment of information about such vector field, including perturbative solutions, invariants, bifurcation properties, stability analysis, numerical integration, integrability results, etc. For instance, see \cite{dlrjc1}-\cite{bs9} and references therein for a sample. The significance of Poisson dynamical systems is due to several reasons. One is that they provide a generalization of classical Hamiltonian systems, allowing not only for odd-dimensional vector fields, but also because Poisson structure matrices admit a great diversity of forms apart from the fixed one associated with the Hamiltonian case. Additionally, an important feature of Poisson systems is that they are not restricted by the use of canonical transformations. In fact, every diffeomorphic change of variables maps a Poisson system into another Poisson system. 

An additional topic relevant for this work is averaging theory. Its starting point can be traced back to Lagrange's study of the three-body problem as a perturbation of the two-body problem, in spite that no formal proof of the validity of the method was given until Fatou's in the XXth century. Soon after, it was the subject of further investigations that led to the establishment of the averaging method as one of the classical tools for the analysis of nonlinear oscillations. Essentially, the underlying idea of this methodology is to approximate the initial system by an averaged (and presumably simpler) version of it, in such a way that the analysis of certain properties of the averaged system will lead to an understanding of the original system. 

The first-order (in the small perturbation parameter $\varepsilon$) averaging theory for studying periodic orbits of differential equations in standard form and in arbitrary finite dimension $n$ is a classical tool for the analysis of vector fields, see \cite{SVM, Verhulst}. In addition, in \cite{Bu-Lli} the averaging theory for differential equations in $\mathbb{R}^n$ up to order $3$ in $\varepsilon$ was presented. In \cite{Ga-Gi}, the averaging theory in $\mathbb{R}^n$ was described in a recursive way to any order in $\varepsilon$, and it was applied to the center problem for planar differential systems. The second-order averaging theory in $\mathbb{R}^n$ was considered in \cite{Bu-Gi-Lli}, where the key tool in that development is the Lyapunov-Schmidt reduction method applied to the translation Poincar\'e-Andronov mapping. In a recent work \cite{Co-Ga-Pro} the explicit expressions of the bifurcation functions up to third order of the averaging theory in $\mathbb{R}^n$ was presented. Also recently, in \cite{Gi-Gr-Ll} the averaging theory was explicitly developed in $\mathbb{R}$ up to an arbitrary order in $\varepsilon$. Some additional recent works devoted to improve and apply the averaging method are \cite{Bu-Ga,Ci-Lli-Te}.

As anticipated, one of the advantages of the Poisson representation is that it allows the application of diverse methods (frequently coming from Hamiltonian dynamics) in a more general context. According to this perspective, it is worth noting that the averaging method \cite{SVM,Verhulst} provides a quantitative relationship between the solutions of some non-autonomous differential systems and the solutions of the averaged (with respect to the independent variable) autonomous differential system. When the system is $T$-periodic, by using the averaging method it is possible to extract some precise information about the existence and number of $T$-periodic solutions as well as their stability. It is worth emphasizing that for the application of averaging theory it is first necessary to express the system in the called standard form: the right-hand side of the system must be sufficiently small, and actually it must be a function of order $\mathcal{O}(\varepsilon)$ when depending on a small real perturbation parameter $\varepsilon$. To express a perturbed system into a standard form by means of changes of variables and a rescaling of the independent variable is in general neither algorithmic nor an easy task. The purpose of this work is to show that the entire approach just described becomes possible for some classes of Poisson systems when they are perturbed under any nonlinear analytic vector field. 

The structure of the article is the following. Section 2 provides a detailed description of the method developed. Such description is divided in three main parts. Thus, Subsection 2.1 is devoted to the description of the Poisson system family of interest together with the general reduction to the Darboux canonical form. In Subsection 2.2 the transformation of such canonical form into the Lagrange standard form is constructed. Taking this as starting point, the application of averaging theory to the analysis of periodic orbits is the issue of Subsection 2.3. Once the method is fully developed, Section 3 is focused on the presentation of three examples. 

\section{Description of the method}

The purpose of the forthcoming development is to generalize the application of the averaging theory and detect periodic orbits in the framework of Poisson systems analytically perturbed. The Poisson systems to be considered are those globally (in a domain) and constructively analytically reducible to an orbitally equivalent planar harmonic oscillator defined on the symplectic leaves. 

\subsection{Reduction to the Darboux canonical form}

Our starting point is an analytic Poisson structure of dimension $n$ and constant rank $2$ in a open and connected set (a domain) $\Omega \subseteq \mathbb{R}^n$ containing the origin. In terms of local coordinates we have a vector field of the form:
\begin{equation}
\label{poisson-aver-5}
	\frac{\mbox{\rm d}x}{\mbox{\rm d}t} = {\cal J}(x) \cdot \nabla H (x).
\end{equation}
Here $H(x)$ is the analytic Hamiltonian function, and the structure matrix is assumed to have the form ${\cal J}(x) = I(x) {\cal J}_0(x)$, where ${\cal J}_0(x)$ is any analytic rank-2 structure matrix defined in $\Omega$, and  $I : \Omega \to \mathbb{R}$ is any analytic nonvanishing first integral (hence it has the functional form $I (x) \equiv \mu(D_3(x), \ldots , D_n(x),H(x))$, where $\{ D_3(x), \ldots , D_n(x) \}$ is a complete set of functionally independent Casimir invariants of the system, see \cite{bs7,bs6} for further details).

\begin{theorem}
\label{ramones}
Let us consider the Poisson system (\ref{poisson-aver-5}) with Hamiltonian of the form
\begin{equation}\label{Poisson-Aver-1}
    H(x) = \frac{1}{2} [x_1^2 h_1^2(x) + x_2^2 h_2^2(x)]
\end{equation}
where $h_i$ are analytic functions in $\Omega$ satisfying $h_i(0) = 1$ for $i=1,2$. Assume in addition that there exist a complete set of independent Casimir invariants of the form
\begin{equation}\label{Poisson-Aver-2}
D_j(x) = x_j + \phi_j(x)
\end{equation}
where $\phi_j(x)$ are analytic functions in $\Omega$ and such that $\phi_j(0) = 0$ and $\nabla \phi_j(0) = 0$ for $j=3, \ldots, n$. Then system (\ref{poisson-aver-5}) is analytically orbital equivalent to a linear Darboux canonical form in a domain $U \subseteq \Omega$ containing the origin. Moreover, such reduction can be constructively determined.
\end{theorem}
\begin{proof}
Under our hypotheses, the transformation $x \mapsto \Phi(x) = y$ defined by
\begin{equation}\label{Poisson-Aver-3}
    \left\{ \begin{array}{rcl}
    y_i(x) & = & x_i h_i(x) \;\: , \;\:\;\: i = 1,2 \\
    y_j(x) & = & D_j(x) \;\: , \;\:\;\: j = 3, \ldots ,n
    \end{array} \right.
\end{equation}
is an analytic diffeomorphism in a domain $U \subseteq \Omega$ containing the origin whose inverse $\Phi^{-1}(y) = \mathcal{I}_n y + \cdots$ has a linear part being the identity matrix of order $n$. More precisely $U = \Phi^{-1}(\Phi(\Omega))$.

Combining the ideas of \cite{bs7,bs6} we shall develop a generalized procedure such that our Poisson system can be reduced globally in $U$ to a one degree of freedom Hamiltonian system and the Darboux canonical form is accomplished globally and diffeomorphically in $U$. This can be done as follows: in the new coordinate system $(y_1, \ldots ,y_n)$ we arrive to the new Poisson system
\[
	\frac{\mbox{\rm d}y}{\mbox{\rm d}t} = {\cal J}^*(y) \cdot \nabla H^*(y)
\]
where $H^*(y)=H \circ \Phi^{-1}(y)=(y_1^2+y_2^2)/2$, and
\[
    {\cal J}^*(y)= I \left( \Phi^{-1}(y) \right) \, \eta (y) \cdot {\cal J}_D.
\]
Here $\eta (y) = \left[ \left. (\nabla_x y_1(x))^T \cdot {\cal J}(x) \cdot ( \nabla_x y_2(x)) \right] \right| _{\Phi^{-1}(y)}$ where by construction it is $\eta (y) \neq 0$ in $U$ as a direct consequence of the preservation of the structure matrix (constant) rank under any diffeomorphic transformation; and in addition
\begin{equation}
\label{jdnd}
	{\cal J}_D \equiv
	\left( \begin{array}{cc} 0 & 1 \\ -1 & 0 \end{array} \right)
	\oplus {\cal O}_{n-2} =
\left( \begin{array}{cccc}
      0  & 1 & \vline & \mbox{} \\
      -1 & 0 & \vline & \mbox{} \\ \hline
      \mbox{} & \mbox{} & \vline & {\cal O}_{n-2}
      \end{array} \right)
\end{equation}
is the Darboux canonical form matrix for the rank-2 case, where ${\cal O}_{n-2}$ denotes the null square matrix of order $n-2$.

Accordingly, in the next reduction step we define the new time $\tau$ in terms of the time reparametrization
\begin{equation}
\label{jr2darbntt}
	\mbox{\rm d} \tau = I \left( \Phi^{-1}(y) \right) \, \eta (y) \, \mbox{\rm d} t
\end{equation}
After both transformations, the outcome is the linear Poisson system
\begin{equation}
\label{poisson-aver-6}
	\frac{\mbox{\rm d}y}{\mbox{\rm d} \tau} = {\cal J}_D \cdot \nabla H^*(y)
\end{equation}
where the Darboux canonical form is therefore constructed. In other words, both Poisson systems (\ref{poisson-aver-5}) and (\ref{poisson-aver-6}) are orbitally equivalent in $U$, and consequently the former is orbitally linearized.
\end{proof}

\begin{remark}
{\rm In the previous procedure, the resulting Hamiltonian function corresponds to a planar harmonic oscillator $H^*(y)=(y_1^2+y_2^2)/2$ defined on the symplectic leaves. Note that in our case the structure matrix ${\cal J}_D$ coincides with the real Jordan canonical form of a real matrix having associated eigenvalues $\pm i$ with $i^2 = -1$ as well as the zero eigenvalue with algebraic multiplicity $n-2$. In particular, this implies that the initial Poisson system (\ref{poisson-aver-5}) must have a couple of nonzero pure imaginary eigenvalues $\pm i \omega$ as well as the zero eigenvalue with algebraic multiplicity $n-2$. Obviously, the eigenvalues $\pm i \omega \neq 0$ are rescaled to their final value $\pm i$ after the time reparametrization (\ref{jr2darbntt}).}
\end{remark}

\subsection{Perturbation and reduction to the Lagrange standard form}

We consider now the analytical perturbations of the initial Poisson system
\begin{equation}
\label{poisson-aver-1}
	\frac{\mbox{\rm d}x}{\mbox{\rm d}t} = {\cal J}(x) \cdot \nabla H (x) + \varepsilon F(x; \varepsilon)
\end{equation}
where $\varepsilon \neq 0$ is a small perturbation real parameter and $F$ is an analytic vector field in $\Omega$ depending analytically on the parameter $\varepsilon$ and satisfying $F(0; \varepsilon) = 0$ and $\nabla_x F(0; \varepsilon) = 0$. The purpose of the perturbations will be to analyze the bifurcation phenomena of periodic orbits via the averaging theory. We shall see that the reduction of Theorem \ref{ramones} can be used in order to write in suitable coordinates the perturbed system (\ref{poisson-aver-1}) into the so-called Lagrange standard form.

\begin{theorem}\label{ramones2}
Consider the perturbation (\ref{poisson-aver-1}) of the Poisson system (\ref{poisson-aver-5}) in $\Omega$. Then, after applying the same reduction performed in Theorem \ref{ramones}, followed by an $n$-dimensional cylindrical transformation, system (\ref{poisson-aver-1}) is constructively transformed into a Lagrange standard form.
\end{theorem}
\begin{proof}
Notice that any perturbation field $F$ preserves both the linear part and the analyticity in $U$ of the system obtained after the transformation $x \mapsto \Phi(x) = y$ defined by (\ref{Poisson-Aver-3}) and the time reparametrization (\ref{jr2darbntt}) leading to the Darboux canonical form for the unperturbed system. We thus obtain that (\ref{poisson-aver-1}) becomes the analytic system
\begin{equation}
\label{poisson-aver-4}
   \frac{\mbox{\rm d}y}{\mbox{\rm d} \tau} = {\cal J}_D \cdot \nabla H^*(y) + \varepsilon F^*(y; \varepsilon)
\end{equation}
defined in $\Phi(U)$. Now we perform the following change to $n$-dimensional cylindrical coordinates
$$
y \mapsto \Psi(y) = (\theta, r, z) \in \mathbb{S}^1 \times U^*
$$
defined by $y_1 = r \cos\theta$, $y_2 = r \sin\theta$ and $z_j = y_j$ for $j=3, \ldots, n$. Here $\mathbb{S}^1 = \mathbb{R} / 2 \pi \mathbb{Z}$ and $U^* = \{ (r, z) \in \mathbb{R}^+ \times \mathbb{R}^{n-2} :  (r \cos\theta, r \sin\theta, z) \in \Phi(U) \mbox{ for all } \theta \in \mathbb{S}^1 \}$. In these coordinates the system becomes
\begin{eqnarray}
\dot{r} &=& \varepsilon \, G_1^*(\theta, r, z; \varepsilon) \ , \nonumber \\
\dot{\theta} &=& -1 + \frac{\varepsilon}{r} G_2^*(\theta, r, z; \varepsilon)   \ , \label{poisson-aver-2} \\
\dot{z}_j &=& \varepsilon \, G_j^*(\theta, r, z; \varepsilon)   \ , \ j=3, \ldots, n \ , \nonumber
\end{eqnarray}
where
\begin{eqnarray*}
G_1^*(\theta, r, z; \varepsilon) &=& \cos\theta \, F_1^*(\Psi^{-1}(\theta, r, z); \varepsilon) + \sin\theta \, F_2^*(\Psi^{-1}(\theta, r, z); \varepsilon) \ ,  \\
G_2^*(\theta, r, z; \varepsilon) &=&  \cos\theta \, F_2^*(\Psi^{-1}(\theta, r, z); \varepsilon) - \sin\theta \, F_1^*(\Psi^{-1}(\theta, r, z); \varepsilon) \ ,  \\
G_j^*(\theta, r, z; \varepsilon) &=& F_j^*(\Psi^{-1}(\theta, r, z); \varepsilon) \ , \ j=3, \ldots, n \ .
\end{eqnarray*}
Notice that this system is only well defined for $r >0$. Moreover, in this region, since for sufficiently small $\varepsilon$ we have
$\dot{\theta} < 0$ in an arbitrarily large ball centered at the origin, we can
rewrite the differential system (\ref{poisson-aver-2}) in such ball into the form
\begin{equation}\label{poisson-aver-3}
\frac{d r}{d \theta}  = \varepsilon  \, G_1(\theta, r, z; \varepsilon) \ , \ \frac{d z}{d \theta}  = \varepsilon \, G_2(\theta, r, z; \varepsilon) ,
\end{equation}
by taking $\theta$ as the new independent variable and using an obvious vectorial notation for the $z$ variables. It is worth emphasizing that any $2 \pi$--periodic solution of (\ref{poisson-aver-3}) corresponds biunivocally with a periodic orbit of (\ref{poisson-aver-1}) in $U$.

System (\ref{poisson-aver-3}) is $2 \pi$--periodic in variable $\theta$ and
is in the Lagrange standard form.
\end{proof}

\subsection{Periodic orbits via averaging theory}

Now we shall present the basic results from averaging theory that we shall need to apply the theory in the forthcoming sections. The classical averaging theory, that is, the first order in $\varepsilon$ averaging theory, is presented for example in Theorems 11.5 and 11.6 of Verhulst \cite{Verhulst}. See also \cite{SVM} for more details. Concerning averaging theory up to third order in $\varepsilon$, the reader is referred to \cite{Bu-Lli}.

The reduction to the standard form (\ref{poisson-aver-3}) performed in Theorem \ref{ramones2} can now be used in order to apply the averaging theory. More precisely, we define the vector function $G(\theta, r, z; \varepsilon) = (G_1(\theta, r, z; \varepsilon), G_2(\theta, r, z; \varepsilon))$ and since $G$ is  analytic in the parameter $\varepsilon$, we can expand in Taylor series the function $G(\theta, r, z; \varepsilon) = \sum_{k \geq 0} g_{k}(\theta, r, z) \, \varepsilon^k$.

The averaging theory up to second order shall be used in what follows, although there are no restrictions in order to increase the perturbation order if required. The auxiliary function is defined as:
\[
    \rho(\theta , r,z) = D_{(r,z)} g_{0}(\theta, r, z) \, \int_0^\theta g_{0}(s, r, z) \, ds + g_{1}(\theta, r, z)
\]
where $D_{(r,z)} f(\theta, r, z)$ denotes the Jacobian matrix with respect to the derivation variables $(r,z)$ for any differentiable function $f$. The following bifurcation functions can now be constructed by averaging with respect to the angular variable $\theta$:
\begin{equation}\label{poisson-aver-12}
\bar{g}_{0}(r, z) = \frac{1}{2 \pi} \int_0^{2 \pi} g_{0}(\theta, r, z) \, d \theta \ , \ \ \ \bar{\rho}(r, z) = \frac{1}{2 \pi} \int_0^{2 \pi} \rho(\theta , r,z) \, d \theta \ .
\end{equation}
Then, up to second order of perturbation we have:
\begin{itemize}
\item[(i)] If $\bar{g}_{0}(r, z) \not\equiv 0$, then for each simple zero $(r_0, z_0) \in U^*$ with $r_0>0$ of $\bar{g}_{0}(r, z)$ and for all $|\varepsilon| > 0$ sufficiently small, there exists a $2 \pi-$periodic solution $\xi(\theta ;\varepsilon)$ of equation (\ref{poisson-aver-3}) such that $\xi(0 ;\varepsilon) \rightarrow (r_0, z_0)$ as $\varepsilon \rightarrow 0$. Moreover, if all the eigenvalues of $D_{(r,z)} \bar{g}_{0}(r_0, z_0)$  have negative real part, the corresponding periodic orbit $\xi(\theta ;\varepsilon)$ is asymptotically stable for $\varepsilon$ sufficiently small. Otherwise, if one of such eigenvalues has positive real part, then $\xi(\theta ;\varepsilon)$ is unstable.

\item[(ii)] If $\bar{g}_{0}(r, z) \equiv 0$ and $\bar{\rho}(r, z) \not\equiv 0$, then for each simple zero $(r_0, z_0) \in U^*$ with $r_0>0$ of $\bar{\rho}(r, z)$ and for all $|\varepsilon| > 0$ sufficiently small, there exists a $2 \pi-$periodic solution $\xi(\theta ;\varepsilon)$
of equation (\ref{poisson-aver-3}) such that $\xi(0 ;\varepsilon) \rightarrow (r_0, z_0)$ as $\varepsilon \rightarrow 0$.
\end{itemize}

\begin{remark}\label{explain}
{\rm As we mentioned, when $\bar{g}_{0}(r, z) \not\equiv 0$, a simple zero $(r_0, z_0) \in U^*$ with $r_0>0$ of $\bar{g}_{0}$ corresponds, for $|\varepsilon| > 0$ sufficiently small, with a $2 \pi-$periodic solution of equation (\ref{poisson-aver-3}) with initial condition tending to $(r_0, z_0)$ as $\varepsilon \rightarrow 0$. Clearly, this periodic orbit is in correspondence with a periodic orbit of (\ref{poisson-aver-4}) such that its initial condition tends to a point $y_0 \in \Phi(U)$ as $\varepsilon \rightarrow 0$ which in turn corresponds with a periodic orbit of (\ref{poisson-aver-1}) having an initial condition which tends to a point $x_0 \in U$ as $\varepsilon \rightarrow 0$.  }
\end{remark}

\begin{remark}\label{putada}
{\rm As described, after some computations it is possible to find the function $g_{0}(\theta, r, z)$ in terms of $F \circ \Phi^{-1} \circ \Psi^{-1}(\theta, r, z)$ needed to arrive at the expression of the first bifurcation function $\bar{g}_{0}(r, z) = \frac{1}{2 \pi} \int_0^{2 \pi} g_{0}(\theta, r, z) \, d \theta$ as defined in (\ref{poisson-aver-12}). It is worth noting at this point that a common computational difficulty may arise, as far as sometimes the former quadrature is not amenable to a closed form evaluation, even for simple perturbation fields $F$. In those cases, a local analysis of the zeroes of $\bar{g}_{0}$ near $(r, z) =(0,0)$ can be performed with the aim of computing an upper bound of the maximum number of small amplitude periodic orbits. For an instance, see Section \ref{examples}. }
\end{remark}

In what follows, the previous theory is illustrated by means of some applied examples.

\section{Examples}\label{examples}

\subsection{Harmonic oscillator with a potential}

We shall now study the particular case $n=3$, $\Omega = \mathbb{R}^3$, with invariant Casimir $D(x) = x_3$ and Hamiltonian $H(x) = \frac{1}{2} x_1^2 + \frac{1}{2} x_2^2 + V(x)$ with potential function $V(x) = \frac{1}{2} x_2^2 h(x_1, x_3) [2 + h(x_1, x_3)]$. Notice that in this case we have, according with (\ref{Poisson-Aver-1}) and (\ref{Poisson-Aver-2}): $\phi \equiv 0$, $h_1 \equiv 1$ and $h_2(x) = 1 + h(x_1, x_3)$. Therefore the diffeomorphism $x \mapsto \Phi(x) = y$ defined in (\ref{Poisson-Aver-3}) on $U$ is given by $\Phi(x) = (x_1, x_2 (1 + h(x_1, x_3)), x_3)$. Its inverse is
$$
\Phi^{-1}(y) = \left(y_1, \frac{y_2}{1 + h(y_1, y_3)}, y_3 \right)
$$
and the rescaling factor is $\eta(y) = 1 + h(y_1, y_3)$.

Let us now take the perturbed system (\ref{poisson-aver-1}) with a perturbation field $F(x; \varepsilon)$ independent of $\varepsilon$ for simplicity, and having an homogeneous quadratic polynomial in $x$ as nonlinearities. Therefore we will assume
\begin{equation}\label{poisson-aver-10}
F(x; \varepsilon) = \left(\sum_{i+j+k=2} a_{ijk} x_1^i x_2^j x_3^k, \sum_{i+j+k=2} b_{ijk} x_1^i x_2^j x_3^k, \sum_{i+j+k=2} c_{ijk} x_1^i x_2^j x_3^k \right) .
\end{equation}

For the sake of clarity, we shall only consider some specific cases of the computations. The next result is about the case $h(x_1, x_3) = x_3$ and $c_{200} = 0$.

\begin{proposition}
Consider the unperturbed Poisson system (\ref{poisson-aver-5}) in $\mathbb{R}^3$ with Hamiltonian $H(x) = \frac{1}{2} x_1^2 + \frac{1}{2} x_2^2 + V(x)$, potential $V(x) = \frac{1}{2} x_2^2 x_3 (2 +  x_3)$ and structure matrix
$$
{\cal J}(x) = \left( \begin{array}{ccc} 0 & 1  & 0 \\ -1 & 0 & 0 \\ 0 & 0 & 0 \end{array} \right) .
$$
Now we take the perturbed system (\ref{poisson-aver-1}) with $F(x)$ given by (\ref{poisson-aver-10}) with $c_{200} = 0$. Define the unbounded domain $U = \mathbb{R}^2 \times (-1, \infty)$. Assuming that $c_{020} \neq 0$ and $c_{002} - 2 (a_{101} + b_{011}) \neq 0$, for $|\varepsilon| \neq 0$ sufficiently small system (\ref{poisson-aver-1}) has exactly $m$ periodic orbits $\xi_m(t ;\varepsilon)$ such that $\xi_m(0 ;\varepsilon)$ tends to a point in $U$ as $\varepsilon \rightarrow 0$ under the following parameter restrictions:
\begin{enumerate}
\item[(i)] $m=1$ if $c_{002} (a_{101} + b_{011}) \neq 0$, $c_{002}/c_{020} < 0$ and $c_{002} \in \mathbb{R} \backslash [0, k]$ if $k > 0$ or $c_{002} \in \mathbb{R} \backslash [k, 0]$ if $k < 0$ where $k = 2 (a_{101} + b_{011})$.

\item[(ii)] $m=0$ in the complementary parameter restrictions of (i).
\end{enumerate}
If additionally $c_{020} = c_{002}  =  b_{002} = b_{200} =  a_{110} = b_{020} = 0$, $a_{101} =- b_{011}$, $a_{200} = 2 c_{101}$ we have:
\begin{enumerate}
\item[(iii)] $m=0$ if $c_{011} = 0$ and $b_{011} c_{110} \neq 0$.
\end{enumerate}
\end{proposition}
\begin{proof}
Note first that the form of the structure matrix ${\cal J}(x)$ implies that the invariant Casimir is $D(x) = x_3$. Second, since $h(x_1, x_3) = x_3$ we construct the diffeomorphism $\Phi(x) = (x_1, x_2(1+x_3), x_3)$ with inverse $\Phi^{-1}(y) = (y_1, y_2 / (1+y_3), y_3)$ in the domain $U = \mathbb{R}^2 \times (-1, \infty)$. Taking $c_{200} = 0$ and after some computations, the first bifurcation function $\bar{g}_{0}(r, z) = (\bar{g}_{01}(r, z)), \bar{g}_{02}(r, z)$ has the components
\begin{eqnarray}
\bar{g}_{01}(r, z) &=& -  \frac{r[3 c_{020} r^2 + 4 z(1+z^2) \{ a_{101} + b_{011} +(a_{101} + b_{011} + c_{002}) z \} ]}{8 (1+z)^4} \ , \nonumber \\
\bar{g}_{02}(r, z) &=& -\frac{c_{020} r^2 + 2 c_{002} z^2(1+z)^2}{2 (1+z)^3} . \label{poisson-aver-11}
\end{eqnarray}
Now we will find zeros $(r_0, z_0) \in U^*$ of $\bar{g}_{0}(r, z)$. In particular $r_0>0$ and $z_0 \in (-1, \infty)$. The resultant $R(z)$ with respect to $r$ of the polynomials in the numerators of the components of $\bar{g}_{0}(r, z)$ is
$$
R(z) = -8 c_{002} c_{020}^2 z^4 (1 + z)^6 (c_{002} z - 2 a_{101} (1 + z) - 2 b_{011} (1 + z))^2
$$
whose roots are $0$, $-1$ and $z_0$ (its precise expression will be displayed later). The root $0$ is not valid because its associated value of $r_0$ is $r_0 = 0$. The root $-1$ is not valid because $-1 \not\in (-1, \infty)$ according to the definition of $U$.

To summarize, we have proved that $\bar{g}_{0}(r, z)$ possesses either none or exactly one root $(r_0, z_0)$ with $r_0 > 0$ and $z_0 \in (-1, \infty)$ depending on some parameter restrictions of the perturbation field $F$. More precisely, the root is
$$
(r_0, z_0) = \left(  \sqrt{- \frac{8 (a_{101} + b_{011})^2 c_{002}^3}{(c_{002} -2 (a_{101} + b_{011}))^4 c_{020}}} , \frac{2 (a_{101} + b_{011})}{c_{002} - 2 (a_{101} + b_{011})} \right)
$$
which appears under the following restrictions: (i) $c_{002} (a_{101} + b_{011}) \neq 0$ and $c_{002}/c_{020} < 0$ to have $r_0 > 0$; (ii) To ensure that $z_0 \in (-1, \infty)$ we define $k = 2 (a_{101} + b_{011})$ and we must have either $c_{002} \in \mathbb{R} \backslash [0, k]$ if $k > 0$ or $c_{002} \in \mathbb{R} \backslash [k, 0]$ if $k < 0$; (iii) in order to avoid zeroes in denominators $c_{020} \neq 0$ and $c_{002} - 2 (a_{101} + b_{011}) \neq 0$. Using the averaging theory at first order we conclude with statements (i) and (ii).

In order to prove statement (iii) equation (\ref{poisson-aver-11}) is considered in the case $c_{020} = c_{002} = a_{101} + b_{011} = 0$ so that $\bar{g}_{0}(r, z) \equiv 0$. Then we must go to second order of bifurcation. After some computations it is found that the second bifurcation function $\bar{\rho}(r, z)$ given by (\ref{poisson-aver-12}) is
$$
\bar{\rho}(r, z) = \left( \frac{r P(r^2,z)}{8 (1+z)^6}, \frac{-z Q(r^2,z)}{2 (1+z)^5} \right)
$$
where $P$ and $Q$ are polynomials. For simplicity in the following computations we take $b_{002} = b_{200} =  a_{110} = b_{020} = c_{011} = 0$, $a_{200} = 2 c_{101}$ and $b_{011} c_{110} \neq 0$. Then $P(r^2, z) = -b_{011} c_{110} r^2 (1 + z)^2 (-1 + 2 z)$ and $Q(r^2, z)= b_{011} c_{110} r^2 (1 + z)^2$ and therefore $\bar{\rho}$ does not have any root $(r_0, z_0)$ with $r_0 > 0$ and $z_0 \in (-1, \infty)$. This proves statement (iii).
\end{proof}

\subsection{A zero-Hopf singularity analytic normal form}

With $\Omega = \mathbb{R}^3$ and the invariant Casimir $D(x) = x_3 + \phi(x_1, x_2)$ we obtain the following structure matrix
\begin{equation}\label{structure-NF}
{\cal J}(x) = \left( \begin{array}{ccc} 0 & 1  & - \partial_{x_2}\phi(x) \\ -1 & 0 & \partial_{x_1}\phi(x) \\ \partial_{x_2}\phi(x) & - \partial_{x_1}\phi(x) & 0 \end{array} \right) .
\end{equation}
According to the previous theory, we shall use the Hamiltonian $H(x) = \frac{1}{2} (x_1^2 + x_2^2)$. This is a case in which the inverse of the diffeomorphism $\Phi$ given by (\ref{Poisson-Aver-3}) on $\Omega$ is explicitly and globally invertible in $\Omega$. More precisely $\Phi(x) = (x_1, x_2, x_3 + \phi(x_1,x_2))$ whose inverse becomes $\Phi^{-1}(y) =(y_1, y_2, y_3- \phi(y_1, y_2))$ and we have a constant rescaling $\eta(y) = 1$.

It is known (see for instance \cite{M}) that our Poisson system (\ref{poisson-aver-5}) associated with these structure matrix ${\cal J}$ and Hamiltonian $H$ corresponds to a special case of an analytic normal form of the zero-Hopf singularity at the origin in the particular case that $\phi(x,y) = P(x^2+y^2)$ with $P$ an analytic function at the origin with $P(0)=0$.

\begin{proposition}
Consider the unperturbed Poisson system (\ref{poisson-aver-5}) in $\mathbb{R}^3$ with Hamiltonian $H(x) = \frac{1}{2} (x_1^2 + x_2^2)$ and structure matrix (\ref{structure-NF}) where $\phi(x_1, x_2) = P(x_1^2 + x_2^2)$ is an analytic function with $P(0)=0$. Consider a perturbation of it as in (\ref{poisson-aver-1}) with any analytic perturbation field
\begin{equation}\label{laF}
F(x; \varepsilon) = \left(\sum_{i+j+k \geq 2} a_{ijk} x_1^i x_2^j x_3^k, \sum_{i+j+k \geq 2} b_{ijk} x_1^i x_2^j x_3^k, \sum_{i+j+k \geq 2} c_{ijk} x_1^i x_2^j x_3^k \right) .
\end{equation}
Assume that at least one of the following three conditions is not fulfilled: (I) $a_{ijk} = 0$ if $i$ is odd and $j$ is even; (II)  $b_{ijk} = 0$ if $i$ is even and $j$ is odd; (III) $c_{ijk} = 0$ if both $i$ and $j$ are even. Let $m$ be the number of nontrivial periodic orbits in $\mathbb{R}^3$ of (\ref{poisson-aver-1}) with $|\varepsilon| \neq 0$ sufficiently small. Then the following holds.
\begin{enumerate}
\item[(i)] If $F$ is an homogeneous polynomial vector field then $m=0$.

\item[(ii)] If $F$ is polynomial of degree 3 then $m \in \{0, 1, 2 \}$ and all the possibilities are realizable.
\end{enumerate}
\end{proposition}
\begin{proof}
Let us first recall that the map $\Phi$ of (\ref{Poisson-Aver-3}) is in this case a global diffeomorphism in $U = \mathbb{R}^3$. Take any analytic perturbation field $F(x; \varepsilon) = (F_1(x), F_2(x), F_3(x))$ such that $F_i(x; \varepsilon) = O(\|x\|^2)$.  Computing the first order bifurcation function $\bar{g}_{0}$ as defined in (\ref{poisson-aver-12}) we find that $\bar{g}_{0}(r, z) = (A(r, z), B(r, z) + 2 r  P'(r^2) A(r, z))$ where
\begin{eqnarray*}
A(r, z) &=& - \frac{1}{2 \pi} \int_0^{2 \pi} \cos\theta \, F_1 \circ \gamma(\theta, r, z) + \sin\theta \, F_2 \circ \gamma(\theta, r, z) \, d \theta , \\
B(r, z) &=& - \frac{1}{2 \pi} \int_0^{2 \pi} F_3 \circ \gamma(\theta, r, z)  \, d \theta ,
\end{eqnarray*}
and $\gamma(\theta, r, z) = (r \cos\theta, r \sin\theta, z-P(r^2))$. Clearly the zeroes of $\bar{g}_{0}$ coincide with the zeroes of $\bar{G}_{0}(r,z) = (A(r,z), B(r,z))$. Additionally the Jacobian is
$$
\det(D \bar{g}_{0}) = \frac{\partial B}{\partial z} \frac{\partial A}{\partial r} - \frac{\partial A}{\partial z} \left( 2 A (P' + 2 r^2 P'') + \frac{\partial B}{\partial r} \right)
$$
which implies that the simple zeroes $(r_0, z_0)$ of $\bar{g}_{0}$ also coincide with the simple zeroes of $\bar{G}_{0}(r,z)$.

Using now the Maclaurin expansion (\ref{laF}) we obtain that
$$
A(r,z) = \sum_{i+j+k \geq 2} I_{ijk} \, r^{i+j} (z-P(r^2))^k \ , \  B(r,z) = \sum_{i+j+k \geq 2} J_{ijk} \, r^{i+j} (z-P(r^2))^k \ ,
$$
where
\begin{eqnarray*}
I_{ijk} &=&  - \frac{1}{2 \pi} \int_0^{2 \pi} (a_{ijk} \cos^{i+1}\theta \sin^j\theta + b_{ijk} \cos^{i}\theta \sin^{j+1}\theta)\, d \theta, \\
J_{ijk} &=&  - \frac{1}{2 \pi} c_{ijk} \int_0^{2 \pi} \cos^{i}\theta \sin^j\theta \, d \theta.
\end{eqnarray*}
It is worth mentioning that $\bar{g}_{0} \equiv 0$ if and only if $\bar{G}_{0} \equiv 0$ which is equivalent to $I_{ijk} = J_{ijk} = 0$ for all subindexes. Therefore, $\bar{g}_{0} \equiv 0$ if and only if $a_{ijk} = 0$ if $i$ is odd and $j$ is even, $b_{ijk} = 0$ if $i$ is even and $j$ is odd and $c_{ijk} = 0$ if both $i$ and $j$ are even. Hence, if at least one of the the conditions (I), (II) and (III) of the statement of the proposition  is not fulfilled  then we have $\bar{g}_{0} \not\equiv 0$ and we can use first order averaging theory to study the periodic orbits of the perturbed system (\ref{poisson-aver-1}) in all the phase space $\mathbb{R}^3$.

We use the new variable $w$ defined by $w = z-P(r^2)$. Thus we obtain that $G^\dag(r, w) = \bar{G}_{0}(r, w+P(r^2)) = (A(r, w+P(r^2)), B(r, w+P(r^2)))$ has the expression
$$
G^\dag(r, w) = ({G}^\dag_1(r, w), {G}^\dag_2(r, w)) = \left(\sum_{i+j+k \geq 2} I_{ijk} \, r^{i+j} w^k, \sum_{i+j+k \geq 2} J_{ijk} r^{i+j} w^k \right) .
$$
We shall analyze the simple real zeroes $(r_0, w_0)$ of $G^\dag$ with $r_0>0$ in some cases. Recall that these zeroes are in correspondence with those zeroes $(r_0, z_0)$ of $\bar{g}_{0}$ via $z_0 = w_0 + P(r_0^2)$.

To prove statement (i) we assume from now on that $F$ is any homogeneous polynomial perturbation field degree of $d \geq 2$, namely it has the form $F(x; \varepsilon) = \left(\sum_{i+j+k = d} a_{ijk} x_1^i x_2^j x_3^k, \sum_{i+j+k = d} b_{ijk} x_1^i x_2^j x_3^k, \sum_{i+j+k = d} c_{ijk} x_1^i x_2^j x_3^k \right)$. Then $G^\dag(r, w) =  \left(\sum_{i+j+k = d} I_{ijk} \, r^{i+j} w^k, \sum_{i+j+k = d} J_{ijk} r^{i+j} w^k \right)$ has its components $G^\dag_i$ given by  homogeneous polynomials in $\mathbb{R}[r,w]$ of degree $d$. Therefore we state that ${G}^\dag$ has not simple real zeros $(r_0, w_0)$ with $r_0 > 0$. Such claim follows after taking into account that (due to homogeneity) there is a unique factorization $G^\dag_i = \prod_{j=1}^d L_{ij}$ where $L_{ij}(r, w) \in \mathbb{C}[r,w]$ are linear polynomials, hence $L_{ij}(0, 0) = 0$. In particular, the only real zeroes $(r_0, w_0) \in \mathbb{R}^2$ of ${G}^\dag$ are either $(r_0, w_0) = (0,0)$ (corresponding to the intersection of two real lines $L_{1j} = 0$ and $L_{2 k} = 0$) or they are multiple of each other (belonging to the intersection of two real coincident lines $L_{1j} = L_{2 k} = 0$).

Going back we find that $\bar{g}_{0}$ has no real simple root $(r_0, z_0)$ with $r_0 > 0$. Hence, since $U = \mathbb{R}^3$ by applying first order averaging theory we conclude that, for $|\varepsilon| \neq 0$ sufficiently small, the perturbed system (\ref{poisson-aver-1}) has no periodic orbits in $\mathbb{R}^3$.

Now we shall prove statement (ii). Let $F$ be any admissible polynomial perturbation field of degree 3, hence it is of the form (\ref{laF}) but with $2 \leq i+j+k \leq 3$. We arrive at 
$$
G^\dag(r, w) =  ({G}^\dag_1(r, w), {G}^\dag_2(r, w)) = \left(\sum_{2 \leq i+j+k \leq 3} I_{ijk} \, r^{i+j} w^k, \sum_{2 \leq i+j+k \leq 3} J_{ijk} r^{i+j} w^k \right) .
$$
More precisely one has
$$
G_1^\dag(r, w) = r \Big[ \alpha_1 r^2 +\beta_1 w + \gamma_1 w^2  \Big] , \ G_2^\dag(r, w) = \alpha_2 r^2 + \beta_2 r^2 w + \gamma_2 w^3,
$$
where $\alpha_1 = -(a_{120} + 3 a_{300} + 3 b_{030} + b_{210})/8$, $\beta_1 = -(a_{101} + b_{011}) / 2$, $\gamma_1 = -(a_{102} + b_{012})/2$, $\alpha_2 = -(c_{020} + c_{200})/2$, $\beta_2 =- (c_{021} + c_{201})/2$ and $\gamma_2 0 -(c_{002} + c_{003})$. Observe that the parameters $\alpha_i$, $\beta_j$ and $\gamma_k$ are independent because of the independence of the parameters $a_{ijk}$, $b_{ijk}$ and $c_{ijk}$. Therefore, solving for $r^2$ from $G_1^\dag(r, w) =0$ and inserting it into the equation $G_2^\dag(r, w) =0$ gives $w Q_2(w) = 0$ with $Q_2(w)$ a polynomial of second degree. Note that zeroes $(r_0, w_0)$ of $G^\dag$ with $w_0 = 0$ imply $r_0=0$, hence they are rejected. In summary $w_0$ must be a root of $Q_2$ and, depending on its discriminant, all the possibilities for the number $m$ of simple zeroes are obtained: $m \in \{0,1,2\}$.
\end{proof}

\subsection{Duffing oscillator}\label{sec.duf}

Duffing oscillator is recognized as one of the paradigmatic examples of planar Hamiltonian dynamics \cite{guho}. In what follows we consider the unforced and undamped Duffing oscillator $\ddot{\varphi} + \varphi + \beta \varphi^3 = 0$ with a real stiffness parameter $\beta$. It is well-known that for $\beta >0$, the equation represents a hard spring system, while for $\beta<0$, it corresponds to a soft spring. In order to analyze perturbations in which parameter $\beta$ is no longer a constant, the system will be embedded in a three-dimensional space by defining $x=(x_1,x_2,x_3)=(\varphi, \dot{\varphi},\beta)$. In terms of these variables, the model can be written as a Poisson system in $\Omega = \mathbb{R}^3$
\begin{equation}\label{duffing1}
\frac{\mbox{\rm d}x}{\mbox{\rm d}t} = {\cal J}_D \cdot \nabla H(x)
\end{equation}
of structure matrix ${\cal J}_D$ which is the 3-dimensional version of (\ref{jdnd}), and Hamiltonian function $H(x)$ of the form (\ref{Poisson-Aver-1}) with $h_1^2(x)=1+x_3x_1^2/2$ and $h_2^2(x)=1$. Notice that $D(x) = x_3$ is a Casimir invariant. On each symplectic leaf $\{ D(x) = c \}$ with $c \in \mathbb{R}$ the system has a center at $(x_1, x_2)=(0,0)$ which has an unbounded period annulus if $c \geq 0$ and a period annulus bounded by heteroclinic connections between two saddle points when $c < 0$.

Let us regard now the analytical perturbations of (\ref{duffing1}) as already defined in (\ref{poisson-aver-1}), that is with an analytic perturbation field given by $F(x; \varepsilon) = (F_1(x; \varepsilon), F_2(x; \varepsilon), F_3(x; \varepsilon))$ without constant nor linear terms in the phase variables $x$.

\begin{proposition}
Consider the unperturbed Poisson representation (\ref{duffing1}) in $\Omega = \mathbb{R}^3$ associated with Duffing system. Consider a perturbation of it as in (\ref{poisson-aver-1}) with any analytic perturbation field $F(x; \varepsilon) = (F_1(x; \varepsilon), F_2(x; \varepsilon), F_3(x; \varepsilon))$ such that $F_i(x; \varepsilon) = O(\|x\|^2)$. Taking $\mathbf{0}=(0,0,0; 0) \in \mathbb{R}^3 \times \mathbb{R}$ and $\partial_{n_1 n_2 n_3} \equiv \partial / \partial x_1^{n_1} \partial x_2^{n_2} \partial x_3^{n_3}$, the following quantities are introduced
\begin{eqnarray}
\Delta_{1} & = & - \partial_{003} F_1(\mathbf{0})-\partial_{012} F_2(\mathbf{0})-\partial_{021} F_1(\mathbf{0})-\partial_{030} F_2(\mathbf{0}) -3\partial_{012} F_1(\mathbf{0})   \nonumber \\
& & -2\partial_{111} F_2(\mathbf{0})-\partial_{120} F_1(\mathbf{0})-3\partial_{201} F_1(\mathbf{0}) -\partial_{210} F_2(\mathbf{0})-\partial_{300} F_1(\mathbf{0}), \nonumber  \\
\Delta_{2} & = & - \partial_{002} F_3(\mathbf{0})-\partial_{020} F_3(\mathbf{0})-2\partial_{101} F_3(\mathbf{0})-\partial_{200} F_3(\mathbf{0}). \nonumber
\end{eqnarray}
In the generic case that $\Delta_1 \neq 0$ or $\Delta_2 \neq 0$, then for $|\varepsilon| \neq 0$ sufficiently small there are no periodic solutions in a neighborhood of the origin of (\ref{poisson-aver-1}).
\end{proposition}
\begin{proof}
The unperturbed system (\ref{duffing1}) and its perturbation $F(x; \varepsilon)$ are ready for the application of the procedure described in Theorem \ref{ramones}. The diffeomorphism $\Phi$ defined in (\ref{Poisson-Aver-3}) takes the form $\Phi(x) =(x_1 \sqrt{1 + x_1^2 x_3 / 2}, x_2, x_3)$ in its definition domain $U = \{ (x_1, x_2, x_3) \in \Omega : x_1^2 x_3 + 2 > 0 \}$. The inverse is
$$
\Phi^{-1}(y) = \left( \sqrt{\frac{-1 + \sqrt{1+2 y_1^2 y_3}}{y_3}} \, , \, y_2 \, , \,  y_3 \right).
$$
Moreover the first integral $I$ and the scalar function $\eta(y)$ used in the time rescaling (\ref{jr2darbntt}) now are $I \equiv 1$ and
$$
\eta(y) = \frac{\sqrt{2+4 y_1^2 y_3}}{\sqrt{1 + \sqrt{1+2 y_1^2 y_3}}} .
$$
The corresponding system (\ref{poisson-aver-4}) defined in $\Phi(U)$ becomes as expected
\begin{equation}
\label{duffing2}
\dot{y_1} = y_2 + \varepsilon F^*_1(y; \varepsilon) , \ \dot{y_2} = -y_1 + \varepsilon F^*_2(y; \varepsilon) , \ \dot{y_3} =  \varepsilon F^*_3(y; \varepsilon).
\end{equation}
Now cylindrical coordinates are taken $y \mapsto \Psi(y) = (\theta, r, z)$ with $y_1 = r \cos\theta$, $y_2 = r \sin\theta$ and $y_3 = z$, and we obtain the expression of the corresponding system (\ref{poisson-aver-3}) in Lagrange standard form. In this example, the bifurcation function $\bar{g}_{0}(r, z) = \frac{1}{2 \pi} \int_0^{2 \pi} g_{0}(\theta, r, z) \, d \theta$ as defined in (\ref{poisson-aver-12}) cannot be obtained in closed form. Therefore a local analysis around $(r,z)=(0,0)$ will be performed as it was anticipated in Remark \ref{putada}.

First we check that $\bar{g}_{0}(r, z) = \left( r^3 \hat{g}_1(r, z), r^2 \hat{g}_2(r, z) \right)$ where $\hat{g}_i(0, z) \not\equiv 0$. Therefore the zeroes with $r > 0$ of $\bar{g}_{0}(r, z)$ and those of the function $\hat{g}(r, z) = \left(\hat{g}_1(r, z), \hat{g}_2(r, z) \right)$ coincide. Additionally it is found that $\hat{g}_1(0, 0) = \Delta_1 / 16$ and $\hat{g}_2(0, 0) = \Delta_2 / 4$ where $\Delta_i$ are defined in the statement of the proposition. Clearly, by continuity of $\hat{g}$ at the origin, if $\Delta_1 \neq 0$ or $\Delta_2 \neq 0$ there are not zeroes of $\hat{g}$ in a neighborhood of $(r, z) = (0,0)$ which in turn proves the proposition if we use first-order averaging theory and go back to the original perturbed system (\ref{poisson-aver-1}).
\end{proof}

\mbox{}

\noindent {\bf Acknowledgments.}

\noindent The first author (I.G.) is partially supported by a MICINN grant number MTM2011-22877 and by a CIRIT grant number 2009 SGR 381. The second author (B.H.-B.) would
like to acknowledge the kind hospitality at Lleida University during which part of this work was developed.

\pagebreak

\end{document}